\documentclass[11pt]{article}
\usepackage{fullpage}
\usepackage{url}

\usepackage{graphics}
\usepackage[dvips]{epsfig}

\usepackage{amsmath}
\usepackage{amssymb}
\usepackage{amsfonts}
\usepackage{graphicx}

\newtheorem{theorem}{Theorem}[section]

\newcommand{\qed}{\hfill $\Box$ \bigbreak}
\newenvironment{proof}{\noindent {\bf Proof.}}{\qed}

\newcommand{\remove}[1]{}
\newcommand{\inst}[1]{$^{#1}$}

\begin{document}

\baselineskip  0.2in 
\parskip     0.05in 
\parindent   0.0in 

\title{{\bf Deterministic meeting of sniffing agents in the plane}\footnote{A preliminary version of this paper appeared in the Proc. 23rd International Colloquium
on Structural Information and Communication Complexity (SIROCCO 2016), LNCS 9988.}}

\author{
Samir Elouasbi\inst{1},
Andrzej Pelc\inst{1}$^,$\footnote{Partially supported by NSERC discovery grant and by the Research Chair in Distributed Computing at the Universit\'e du Qu\'{e}bec en Outaouais.}\\
\inst{1} Universit\'{e} du Qu\'{e}bec en Outaouais, Gatineau, Canada.\\
E-mails: \url{elos02@uqo.ca}, \url{pelc@uqo.ca}\\
}

\date{ }
\maketitle

\begin{abstract}

Two mobile agents,
starting at arbitrary, possibly different times from arbitrary locations in the plane, have to meet.
Agents are modeled as discs of diameter 1, and meeting occurs when these discs touch.
Agents have different labels which are integers from the set $\{0,\dots,L-1\}$. Each agent knows $L$ and knows its own label, but not the label of the other agent.
Agents are equipped with compasses and have synchronized clocks. They make a series of moves. Each move specifies the direction and the duration of moving.
This includes a $null$ move which consists in staying inert for some time, or forever.
In a non-null move agents travel at the same constant speed, normalized to 1.

We assume that agents have sensors enabling them to estimate the distance from the other agent (defined as the distance between centers of discs), but not the direction towards it. We consider two models of estimation.
In both models an agent reads its sensor at the moment of its appearance in the plane and then at the end of each move. This reading (together with the previous ones) determines the decision concerning the next move.
In both models the reading of the sensor tells the agent if the other agent is already present. Moreover,
in the {\em monotone model}, each agent can find out, for any two readings in moments $t_1$ and $t_2$, whether the distance from the other agent at time $t_1$
was smaller, equal or larger than at time $t_2$. In the weaker {\em binary model}, each agent can find out, at any reading, whether it is at distance less than $\rho$ or at distance at least $\rho$ from the other agent,
for some real $\rho>1$ unknown to them. Such distance estimation mechanism can be implemented, e.g., using chemical sensors. Each agent emits some chemical substance (scent), and the sensor
of the other agent detects it, i.e., {\em sniffs}. The intensity of the scent decreases with the distance. In the monotone model it is assumed that the sensor is ideally accurate and
can measure any change of intensity. In the binary model it is only assumed that the sensor can detect the scent below some distance (without being able to measure intensity) above which
the scent is too weak to be detected.

We show the impact of the two ways of sensing on the time of meeting, measured from the start of the later agent.
For the monotone model we show an algorithm achieving meeting in time $O(D)$, where $D$ is the initial distance between the agents.
This complexity is optimal. For the binary model we show that, if agents start at distance smaller than $\rho$ (i.e., when they sense each other initially) then
meeting can be guaranteed within time $O(\rho\log L)$, and that this time cannot be improved in general.
Finally we observe that,  if agents start at distance $\alpha\rho$, for some constant $\alpha >1$ in the binary model, then sniffing does not help, i.e., the worst-case optimal meeting time is of the same
order of magnitude as without any sniffing ability.

\end{abstract}

\vfill

\thispagestyle{empty}
\setcounter{page}{0}
\pagebreak

\section{Introduction}
{\bf The background and the problem.}
Two mobile agents,
starting at arbitrary, possibly different times from arbitrary locations in the plane, have to meet.
Agents are modeled as discs of diameter 1, and meeting occurs when these discs touch (i.e., the centers of the agents get at distance 1). This way of formulating the meeting problem in the plane is equivalent to the problem of {\em approach} \cite{DP},
where agents are modeled as points moving in the plane, and the approach is defined as these points getting at distance at most 1 from each other. This is one of the versions of the well-known rendezvous problem
in which two or more agents have to meet in some environment. This problem has been studied in many variations: in the plane and in networks, in the synchronous vs. asynchronous setting, and using deterministic vs. randomized algorithms. In applications, mobile agents may be rescuers trying to find a lost tourist in the mountains, animals trying to find a mate, or mobile robots that have to meet in order to compare previously
collected data.

We are interested in deterministic algorithms for the task of meeting. If agents were anonymous (identical), then, if started simultaneously, they would trace identical trajectories and hence could never meet.
In order to break the symmetry, we assume that agents have different labels which are integers from the set $\{0,\dots,L-1\}$. Each agent knows $L$ and knows its own label which it can use as a parameter in the algorithm that they both execute, but it does not know the label of the other agent.

Agents are equipped with compasses showing the cardinal directions, and have synchronized clocks.
The adversary places each agent at some point of the plane at possibly different times.
The clock of the agent starts at the moment of its appearance in the plane.
Each agent makes a series of moves. A move specifies the direction and the duration of moving.
This includes a $null$ move which consists in staying inert for some time, or forever.
In a non-null move agents travel at the same constant speed, normalized to 1.

We assume that agents have sensors enabling them to estimate the distance from the other agent (defined as the distance between centers of discs), but not the direction towards it. We consider two models of estimation.
In both models an agent reads its sensor at the moment of its appearance in the plane and then at the end of each move. This reading (together with the previous ones) determines the decision concerning the next move.
In both models the reading of the sensor tells the agent if the other agent is already present. Moreover,
in the {\em monotone model} each agent can find out, for any two readings in moments $t_1$ and $t_2$, whether the distance from the other agent at time $t_1$
was smaller, equal or larger than at time $t_2$. In the weaker {\em binary model} each agent can find out, at any reading, whether it is at distance less than $\rho$ or at distance at least $\rho$ from the other agent,
for some real $\rho>1$ unknown to them. (We assume that $\rho>1$ because agents are always at distance larger than 1 before touching, hence $\rho \leq 1$ would be useless for sensing.)   Such distance estimation mechanism can be implemented, e.g., using chemical sensors. Each agent emits some chemical substance (scent), and the sensor
of the other agent detects it, i.e., {\em sniffs}. If  at some time the agent is still alone in the plane, the reading of its sensor is 0.
Otherwise, the reading of the sensor is positive in the monotone model, and in the binary model it is 1 if the distance between the agents is less than $\rho$ and 0 if it is at least $\rho$.
The intensity of the scent decreases with the distance. In the monotone model it is assumed that the sensor is ideally accurate: it
can measure any change of scent intensity and hence can compare distances at any two readings. (The name {\em monotone} comes from the fact that the intensity of the scent accurately sensed by the agent
is a strictly decreasing function of the distance. We do not assume anything else about this function: an agent cannot learn, e.g., the value of its distance to the other agent.) In the binary model it is only assumed that the sensor can detect the scent below some distance (without being able to measure its intensity) above which
the scent is to weak to be detected.

Note that the monotone model is similar to the model of distance-aware agents from \cite{DDKU}. The differences are the environment (networks in the case of  \cite{DDKU} vs. the plane with the use of compasses in our case) and sensing
after each round in  \cite{DDKU} vs. sensing at times decided by the agent in our case. These differences are important and will lead to a more efficient algorithm in our setting.

{\bf Our results.}
We show the impact of the two ways of sensing on the time of meeting, measured from the start of the later agent.
For the monotone model we show an algorithm achieving meeting in time $O(D)$, where $D$ is the initial distance between the agents.
This complexity is optimal.
For the binary model we show that, if agents start at distance smaller than $\rho$ (i.e., when they sense each other initially) then
meeting can be guaranteed within time $O(\rho\log L)$, and that this time cannot be improved in general.
Indeed we show that, for some initial distance less than $\rho$, and for some labels of the agents, time $\Omega(\rho\log L)$ is needed to meet in the binary model.
Finally we observe that  if agents start at distance $\alpha\rho$, for some constant $\alpha >1$ in the binary model, then sniffing does not help, i.e., the worst-case optimal meeting time is of the same
order of magnitude as without any sniffing ability.

Our results show a separation between the two models of sensing accuracy. Suppose that agents start at distance $\rho/3$. Then in the monotone model the optimal meeting time is $\Theta(\rho)$, while
in the binary model it is $\Theta(\rho\log L)$.

{\bf Related work.}
The large literature on rendezvous can be classified according to the mode in which agents move (deterministic or randomized) and
the environment where they move (a network modeled as a graph or a terrain in the plane).
An extensive survey of  randomized rendezvous in various scenarios  can be found in
\cite{alpern02b}, cf. also  \cite{anderson98b,KKPM08}.

Deterministic rendezvous in networks was surveyed in \cite{P}.
In this setting a lot of effort has been dedicated to the study of the feasibility of rendezvous, and to the time required to achieve this task, when feasible, under the synchronous scenario.
For instance, deterministic rendezvous with agents equipped with tokens used to mark nodes was considered, e.g., in~\cite{KKSS}. Time of deterministic rendezvous of agents equipped with unique labels was discussed in \cite{DFKP,TSZ07}. Memory required by the agents to achieve deterministic rendezvous was studied in \cite{BIOKM,FP2} for trees and in  \cite{CKP} for general graphs.
Fault-tolerant rendezvous was studied, e.g., in \cite{CDLP}. In \cite{MP} the authors studied tradeoffs between the time of rendezvous and the total number
of edge traversals by both agents until the meeting. In \cite{DDKU} the authors considered distance-aware agents operating in networks. As mentioned above, this is a model similar to our monotone model.
They showed a rendezvous algorithm polynomial in local parameters of the problem (initial distance between agents, maximum degree and length of the shorter label). They also established a lower bound
for this time which exceeds $\Theta(D)$. This shows a separation between our setting and theirs.

Other works have been devoted to asynchronous rendezvous in networks, cf. e.g., \cite{BCGIL,DPV}, when the agent chooses the edge which it decides to traverse but the adversary controls the speed of the agent. Under this assumption rendezvous
in a node cannot be guaranteed even in very simple graphs, and hence the rendezvous requirement is relaxed to permit the agents to meet inside an edge.

Rendezvous of two or more agents in the plane has been mainly considered in two different settings. In one of them, cf. e.g.,  \cite{CFPS,CGP,fpsw,FSVY}, agents can see the positions
of other agents, and make the decisions based on these observations, usually in an asynchronous way. Another scenario does not allow agents to make any observations.
In \cite{DP} the authors proposed an algorithm for the asynchronous version of the problem of approach in the plane (equivalent to our meeting),
with cost polynomial in the initial distance and in the length of the smaller label.
The results from \cite{BCGIL} are for asynchronous rendezvous in the grid but imply solutions for the approach problem in the plane as well. They use a strong assumption of knowledge of initial positions of the agents in some global system of coordinates, but achieve approach at cost $O(D^2 polylog (D))$, where $D$ is the initial distance.

\section{Terminology and preliminaries}

The direction North -- South is called the {\em vertical} direction, and the direction East -- West is called the {\em horizontal} direction.
We say that agent $a$ is North of agent $b$, if the horizontal line containing the center of agent $a$ is North of the horizontal line containing the center of agent $b$.
The three other expressions (``South of'', ``East of'' and ``West of'') have analogous meaning. We say that agent $a$ is at distance $x$ from $b$ in the vertical direction, if
the horizontal lines containing the centers of the agents are at distance $x$. The distance in the horizontal direction is defined similarly.

Let $(a_1 \dots a_m)$ be the binary representation of the label $\ell$ of an agent. The {\em transformed label} $T(\ell)$ is obtained
by padding the string $(a_1 \dots a_m)$ by a prefix of $\lambda-m$ zeroes, where $\lambda =\lceil \log L \rceil$.
Hence every transformed label has length $\lambda$.
Notice that if the labels are different, then there exists an index for which the corresponding bits of their transformed labels differ (this is not necessarily the case for the binary representations of the original labels, since one of them
might be a prefix of the other). Moreover, if $\ell_1 < \ell_2$, then  $T(\ell_1)$ is lexicographically smaller than $T(\ell_2)$.


\section{The monotone model}

In this section we present an algorithm that accomplishes the meeting in the monotone model in time $O(D)$, where $D$ is the initial distance between the agents.
This is of course optimal, as time $(D-1)/2$ is a lower bound, because the speed of the agents is 1.

The elementary instructions in the monotone model are $read(C)$ and $move(card, x)$. The instruction $read(C)$ results in reading the current  value of the scent intensity sensor into the variable $C$.
Recall that the value of the sensor is 0 if the agent is alone in the plane, and it is some positive real otherwise. The instruction  $move(card, x)$, where $card$ is one of the cardinal directions ($N,E,S,W$)
and $x$ is a positive real, results in moving the agent in the direction $card$ during time $x$. Since the speed of the agent is 1, this means that the agent travels distance $x$ in direction $card$.

At a high level, the idea of the algorithm is the following. In the beginning the agent reads its sensor. If its value is 0, this means that the other agent is not yet in the plane, which enables the agent to break symmetry.
The agent stays inert forever and will be eventually found by the other agent. The other agent must realize that the first agent is inert  and find it by first getting at distance at most 1 in the vertical direction and then getting at distance
at most 1 in the horizontal direction, which results in meeting.

If the initial readings of sensors of both agents are positive, this means that both of them are placed in the plane simultaneously. In this case, the only way to break symmetry is using the (transformed) labels of the agents,
which are different by assumption. The agents must realize that they are in this more difficult situation, and then approach, first in the vertical and then in the horizontal direction.
This is done by moving North and South for the vertical approach (and East and West for the horizontal approach) according to the bits of the transformed label of the agent. At the first bit  where their transformed labels
differ, the symmetry between the agents will be broken, they will realize it by reading their sensors, and then accomplish the approach. In order to keep the time $O(D)$ (and not $O(D+ \log L)$) the moves of the agents corresponding to consecutive bits shrink by a factor of 2 at each consecutive bit.

We now proceed to a detailed description of the algorithm.
We will use the following elementary procedures. The first of them compares the previous reading of the sensor with the current one, and assigns the result of the comparison to the variable $compare$.

\vspace*{0.5cm}

\begin{center}
\fbox{
\begin{minipage}{7cm}

{\bf Procedure} {\tt Test}\\

$P \leftarrow C$\\
$read(C)$\\
{\bf if} $P<C$ {\bf then} $compare \leftarrow larger$\\
{\bf if} $P=C$ {\bf then} $compare \leftarrow equal$\\
{\bf if} $P>C$ {\bf then} $compare \leftarrow smaller$

\end{minipage}
}
\end{center}

\vspace*{0.5cm}

The aim of the next procedure is to approach the other agent either in the vertical or in the horizontal direction by walking in steps of prescribed length $x$.
It is used when the agent already realized that the other agent is  North (resp. South) of it, or it is East (resp. West) of it.
We formulate the procedure for the parameter $card$ which can be equal to $N$, $E$, $S$, or $W$.

\vspace*{0.5cm}

\begin{center}
\fbox{
\begin{minipage}{7cm}

{\bf Procedure} {\tt GetCloser} $(card,x)$\\

{\bf while} $compare = smaller$ {\bf do}\\
\hspace*{1cm}$move(card,x)$\\
\hspace*{1cm}{\tt Test}

\end{minipage}
}
\end{center}

\vspace*{0.5cm}

Our next procedure is called in the case when agents appear simultaneously in the plane, and its aim is to break symmetry between them in this case.
This is done by having the value of the variable $compare$ in Procedure {\tt Test} become different from $equal$. This occurs when agents process the bit
of their transformed labels in which they differ.
Note the two attempts at  getting the value of the variable $compare$ different from $equal$. This is necessary in the special case when
transformed labels of the agents differ in only one bit, say with index $j$, the distance in the vertical direction between the agents is $\frac{1}{2^j}$,
and the agent whose $j$th bit is 1 is South of the agent whose $j$th bit is 0. In this special case, a single attempt at breaking symmetry would go unnoticed:
agents would switch positions in the vertical direction and the value of variable $compare$ would remain $equal$. The procedure is executed by an agent with label $\ell$.

\begin{center}
\fbox{
\begin{minipage}{7cm}

{\bf Procedure} {\tt Dance}\\

$T(\ell) \leftarrow (c_1\dots c_{\lambda})$\\
$i \leftarrow 1$\\
{\bf while} $compare =equal$ {\bf do}\\
\hspace*{1cm}{\bf if} $c_i=1$ {\bf then}\\
\hspace*{2cm}$move(N,\frac{1}{2^i})$\\
\hspace*{2cm}{\tt Test}\\
\hspace*{2cm}{\bf if} $compare =equal$ {\bf then}\\
\hspace*{3cm}$move(N,\frac{1}{2^i})$\\
\hspace*{3cm}{\tt Test}\\
\hspace*{1cm}{\bf else}\\
\hspace*{2cm}$move(S,\frac{1}{2^i})$\\
\hspace*{2cm}{\tt Test}\\
\hspace*{2cm}{\bf if} $compare =equal$ {\bf then}\\
\hspace*{3cm}$move(S,\frac{1}{2^i})$\\
\hspace*{3cm}{\tt Test}\\
\hspace*{1cm}$i \leftarrow i+1$

\end{minipage}
}
\end{center}

We now describe two main procedures of our algorithm which will be called one after another. The aim of the first of them is to get the agents at distance at most 1 in the vertical direction and the aim of the second is to
get the agents at distance at most 1 in the horizontal direction. In Procedure {\tt VerticalApproach}, the later agent, or both agents if they start simultaneously, first realize which of these two cases occurs.
This is done as follows. The agent, call it $a$, moves North by 1. If it decreased the distance from the other agent, it learns that the other agent is inert, and then approaches it by steps of length 1/2,
getting at distance at most 1 in the vertical direction. If it increased the distance, it also learns that the other agent is inert, goes back (i.e. South by 1), and then approaches the other agent by steps of length 1/2,
getting at distance at most 1 in the vertical direction. If the distance did not change, the situation is still unclear: the start may have been simultaneous and the other agent also moved North by 1, or the other agent
may be inert and was at vertical distance 1/2 North before the move of agent $a$. Agent $a$ clarifies this by moving North by 1 again. Agent $a$ could not decrease its distance from $b$ after this move.
If it increased the distance, it goes back and approaches agent $b$ as before.
On the other hand, if the distance did not change again, agent $a$ learns that it is in the simultaneous start situation. It then performs Procedure {\tt Dance}, at the end of which the value of the variable $compare$ is different from $equal$. The $j$-th bit whose processing caused this change is the first bit where the transformed labels of the agents differ. This breaks symmetry. There are two cases. If, at the end of {\tt Dance}, $compare = smaller$, then the agent whose  $j$-th bit  is 1 was South of the other agent {\em before} the last move.
The agents backtrack by a distance of $\frac{1}{2^j}$ and then approach each other: the agent whose $j$-th bit  is 1 going North, and the agent whose $j$-th bit  is 0 going South. If, at the end of {\tt Dance}, $compare = larger$,
then the agent whose  $j$-th bit  is 1 is North of the other agent {\em after} the last move. Now there is no need of backtracking: the agents simply
 approach each other: the agent whose $j$-th bit  is 1 going South, and the agent whose $j$-th bit  is 0 going North.
In both cases, the steps during approach (Procedure {\tt GetCloser}) are now of length 1/4 instead of 1/2, because agents perform approach simultaneously.

\begin{center}
\fbox{
\begin{minipage}{10cm}

{\bf Procedure} {\tt VerticalApproach}\\

$sim \leftarrow$ {\em false}\\
$move(N,1)$\\
{\tt Test}\\
{\bf if} $compare=smaller$ {\bf then} {\tt GetCloser} $(N,\frac{1}{2})$\\
{\bf else}\\
\hspace*{1cm}{\bf if} $compare=larger$ {\bf then}\\
\hspace*{2cm}$move(S,1)$\\
\hspace*{2cm}{\tt GetCloser} $(S,\frac{1}{2})$\\
\hspace*{1cm}{\bf else}\\
\hspace*{2cm}$move(N,1)$\\
\hspace*{2cm}{\tt Test}\\
\hspace*{2cm}{\bf if} $compare=larger$ {\bf then}\\
\hspace*{3cm}$move(S,1)$\\
\hspace*{3cm}{\tt GetCloser} $(S,\frac{1}{2})$\\
\hspace*{2cm}{\bf else} (*$compare=equal$*)\\
\hspace*{3cm}$sim \leftarrow true$\\
\hspace*{3cm}{\tt Dance}\\
\hspace*{3cm}$j \leftarrow i-1$\\
\hspace*{3cm}{\bf if} $compare=smaller$ {\bf then}\\
\hspace*{4cm}{\bf if} $c_j=1$ {\bf then}\\
\hspace*{5cm}$move(S,\frac{1}{2^j})$\\
\hspace*{5cm}{\tt GetCloser} $(N,\frac{1}{4})$\\
\hspace*{4cm}{\bf else}\\
\hspace*{5cm}$move(N,\frac{1}{2^j})$\\
\hspace*{5cm}{\tt GetCloser} $(S,\frac{1}{4})$\\
\hspace*{3cm}{\bf else} (* $compare=larger$ *)\\
\hspace*{4cm}{\bf if} $c_j=1$ {\bf then}\\
\hspace*{5cm}{\tt GetCloser} $(S,\frac{1}{4})$\\
\hspace*{4cm}{\bf else}\\
\hspace*{5cm}{\tt GetCloser} $(N,\frac{1}{4})$

\end{minipage}
}
\end{center}

The following Procedure {\tt HorizontalApproach} will be called after Procedure {\tt VerticalApproach}, and  uses global variables $sim$ and $j$ whose values are set in the above procedure.
Procedure {\tt HorizontalApproach} is simpler because, when it is called, the agent has two important pieces of information. First, it already knows that its vertical distance from the other agent is at most 1, and hence agents
cannot pass each other horizontally without meeting. Second, the agent knows if the start was simultaneous, or if the other agent started first and hence is inert. This information is coded in the boolean variable $sim$,
which is set to $true$ in Procedure {\tt VerticalApproach} if and only if the start was simultaneous. Moreover, if the start was simultaneous, the agent knows already the first bit in which its transformed label differs from that of the other agent: the index of this bit is $j$. Hence, if $sim =$ {\em false}, the agent approaches the other (inert) agent similarly as before, and if $sim=true$, then agents approach each other horizontally by going either East or West,
depending on the value of the $j$th bit of their transformed label, until meeting.

\begin{center}
\fbox{
\begin{minipage}{10cm}

{\bf Procedure} {\tt HorizontalApproach}\\

{\bf if} $sim =$ {\em false} {\bf then}\\
\hspace*{1cm}$move(E,1)$\\
\hspace*{1cm}{\tt Test}\\
\hspace*{1cm}{\bf if} $compare=smaller$ {\bf then} {\tt GetCloser} $(E,1$)\\
\hspace*{1cm}{\bf else}\\
\hspace*{2cm}$move(W,1)$\\
\hspace*{2cm}{\tt GetCloser} $(W,1$)\\
{\bf else}\\
\hspace*{1cm}{\bf if} $c_j=1$ {\bf then}\\
\hspace*{2cm}$move(E,1)$\\
\hspace*{2cm}{\tt Test}\\
\hspace*{2cm}{\bf if} $compare=smaller$ {\bf then} {\tt GetCloser} $(E,1$)\\
\hspace*{2cm}{\bf else}\\
\hspace*{3cm}$move(W,1)$\\
\hspace*{3cm}{\tt GetCloser} $(W,1$)\\
\hspace*{1cm}{\bf else}\\
\hspace*{2cm}$move(W,1)$\\
\hspace*{2cm}{\tt Test}\\
\hspace*{2cm}{\bf if} $compare=smaller$ {\bf then} {\tt GetCloser} $(W,1$)\\
\hspace*{2cm}{\bf else}\\
\hspace*{3cm}$move(E,1)$\\
\hspace*{3cm}{\tt GetCloser} $(E,1$)

\end{minipage}
}
\end{center}

Now our main algorithm for the monotone model can be formulated succinctly as follows.
It is executed by an agent with  label $\ell$ (that intervenes in Procedure {\tt Dance}), with the understanding that when agents touch (i.e., get at distance 1), the algorithm is interrupted, as meeting is then accomplished.
We will prove that this will always occur by the end of the execution of the algorithm.

\begin{center}
\fbox{
\begin{minipage}{7cm}

{\bf Algorithm} {\tt MeetingWithPreciseSensor}\\

$read(C)$\\
{\bf if} $C=0$ {\bf then} stay inert forever\\
{\bf else}\\
\hspace*{1cm} {\tt VerticalApproach}\\
\hspace*{1cm} {\tt HorizontalApproach}

\end{minipage}
}
\end{center}

\vspace*{0.5cm}

\begin{theorem}
The meeting of two agents that are arbitrarily placed in the plane and execute Algorithm {\tt MeetingWithPreciseSensor} in the monotone model, occurs by the end of the execution of this algorithm.
If agents are at initial distance $D$, then the meeting occurs in time $O(D)$ after the appearance of the later agent.
\end{theorem}

\begin{proof}
Consider two cases.

Case 1. Agents do not start simultaneously.

In this case the first agent, call it $b$, gets the initial value 0 of its sensor, and hence stays inert forever. The other agent, call it $a$, moves North by 1.
If its distance from $b$ decreased, agent $b$ was North of $a$ before the move, and it is either still North of $a$ after the move, or it is South of $a$, at distance less than 1/2 in the vertical direction.
In both cases, Procedure {\tt GetCloser} $(N,\frac{1}{2})$ will get the agent $a$ at distance at most 1 from $b$ in the vertical direction. If the distance of $a$ from $b$ increased after the move, then either
$a$ was South of $b$ at distance less than 1/2 in the vertical direction before the move and is North of $b$ after the move, or $a$ was North of $b$ already before the move. In the first case, going back
(i.e., South by 1) will get agent $a$ again South of $b$ at distance less than 1/2 in the vertical direction. Procedure {\tt GetCloser} $(S,\frac{1}{2})$ will result in one step of length 1/2 South, and hence
agent $a$ will finish Procedure {\tt VerticalApproach} at distance less than 1 from agent $b$ in the vertical direction. In the second case, after the move back, agent $a$ will be still North of $b$, and
hence Procedure  {\tt GetCloser} $(S,\frac{1}{2})$ will end up with the agents at vertical distance less than 1.

If after the first move the distance between the agents did not change, agent $b$ was North of agent $a$,
at vertical distance 1/2, before the move of agent $a$, and after the move it is South of agent $a$,
at vertical distance 1/2. Another move North by 1 increases the distance between the agents. Then agent $a$ goes back, and
Procedure  {\tt GetCloser} $(S,\frac{1}{2})$ will end up after two moves, with the agents at vertical distance exactly 1/2.

Hence agent $a$ finishes Procedure {\tt VerticalApproach} at distance at most 1 from agent $b$ in the vertical direction.
Similarly, Procedure {\tt HorizontalApproach} (performed with $sim =$ {\em false})  will get agent $a$ at distance at most 1 from agent $b$ in the horizontal direction, by which time the meeting will occur.
Let $x$ be the initial distance between the agents in the vertical direction, and let $y$ be the initial distance between the agents in the horizontal direction.
Agent $a$ spent time at most $x+4$ moving in the vertical direction, and spent time at most $y+4$ moving in the horizontal direction, hence the time between the start of agent $a$ and the meeting is at most
$x+y+8 \in O(x+y)=O(D)$, where $D$ is the initial distance between the agents.

Case 2. Agents start simultaneously.

The first reading of the sensor is positive for each agent. Since agents start simultaneously, each of them makes a move North by 1, the value of $compare$ is set to $equal$, and they make a second move North.
After these two moves, both agents realize that the start was simultaneous. They set variable $sim$ to $true$ and perform Procedure {\tt Dance}.
Let $j$ be the index of the first bit where the transformed labels of the agents differ.
This procedure ends for both of them after processing the $j$th bit, with the value of variable $compare$ different from $equal$.
If $compare = smaller$, then the agent whose  $j$-th bit  is 1 was South of the other agent {\em before} the last move.
The agents backtrack by a distance of $\frac{1}{2^j}$, in order to restore the situation before the last move, and then approach each other: the agent whose $j$-th bit  is 1 going North, and the agent whose $j$-th bit  is 0 going South. If $compare = larger$,
then the agent whose  $j$-th bit  is 1 is North of the other agent {\em after} the last move. Now there is no need of backtracking, and the agents
 approach each other as follows: the agent whose $j$-th bit  is 1 goes South, and the agent whose $j$-th bit  is 0 goes North.
Similarly as in case 1, both agents finish Procedure {\tt VerticalApproach} at distance at most 1 from each other in the vertical direction.
Procedure {\tt HorizontalApproach} (performed with $sim = true$)  will get both agents at distance at most 1 from each other in the horizontal direction, by which time the meeting will occur.

It remains to estimate the time of meeting. Agents spend time 2 before calling Procedure {\tt Dance}. Procedure {\tt Dance} takes time at most 1, due to the fact that the time of processing a bit
is twice shorter than that of processing the preceding bit.
Let $x$ be the initial distance between the agents in the vertical direction, and let $y$ be the initial distance between the agents in the horizontal direction. The vertical approach
after Procedure {\tt Dance} takes time less than $x+1$, and the horizontal approach takes time less than $y+1$. Hence the total time in the case of simultaneous start is less than
$x+y+5 \in O(x+y)=O(D)$, where $D$ is the initial distance between the agents.
\end{proof}

\section{The binary model}

In this section we present an algorithm that accomplishes the meeting in the binary model in time $O(\rho \log L)$, assuming that agents are initially at distance smaller than $\rho$, i.e., that they can initially sense each other, when both agents are already in the plane. We also show a matching lower bound on the time of the meeting, for some initial positions at distance smaller than $\rho$ and for some labels of the agents.

The elementary instructions in the binary model are $check(C)$ and $move(card, x)$. The instruction $check(C)$ results in reading a single bit into the variable $C$.
This bit is 1, if the other agent is at distance less than $\rho$, and it is 0 otherwise (i.e., if the other agent is either still absent from the plane, or if it is at distance at least $\rho$).
The instruction  $move(card, x)$, where $card$ is one of the cardinal directions ($N,E,S,W$)
and $x$ is a positive real, is identical as in the monotone model: it results in moving the agent in the direction $card$ during time $x$.

At a high level, the idea of the algorithm is the following. In the beginning the agent reads its sensor. If its value is 0, this means that the other agent is not yet in the plane, which enables the agent to break symmetry
as in the monotone model.
The agent stays inert forever and will be eventually found by the other agent.  As explained below, this will be done in a much different way than in the monotone model, as sensing of the other agent is less precise.

If the initial readings of sensors of both agents are 1, this means that both of them are placed in the plane simultaneously. In this case, they break symmetry  using their transformed labels,
which are different by assumption. Initially, an agent has no way of deciding if it was placed in the plane later, or if both agents were placed in the plane simultaneously. Hence, when the bit initially read in $check(C)$
is 1, the agent performs actions that will enable it to {\em lose contact} with the other agent (i.e., to get the reading 0 in $check(C)$), regardless of which of these situations occurs. Losing contact will
always happen when one agent moves and the other stays. This will enable the agents to break symmetry: the agent during whose move the contact was lost will find the other agent that becomes inert (in the case of non-simultaneous start, the agent that will perform the finding
is the later agent, and the earlier agent is inert from the start).
Once symmetry is broken, the moving agent accomplishes the meeting using the binary readings of its sensor and geometric properties of the plane.

We now give a detailed description of the algorithm. The aim of the first procedure is losing contact between the agents, i.e., having both of them (in the case of simultaneous start), or the later agent
(in the case of non-simultaneous start) get the reading 0 in $check(C)$). This is done by going North or staying inert, for increasing periods of time, according to the bits of the transformed label:
when the bit is 1, the agent moves, when it is 0, it stays inert.
It will be proved that agents eventually get at distance at least $\rho$, which they will realize by reading their sensors. Losing contact occurs when one agent moves North and the other agent is inert.
This happens when agents, in the case of simultaneous start (or only the later agent, otherwise) process the $j$th bit of their transformed label.
When contact is lost, the agent during whose move this occurred, i.e., the agent for which $c_j=1$,  goes back South, this time by small steps, until contact is regained (the reading of $C$ is 1 again). This is done to determine the point when contact has been lost, with sufficient precision.
The procedure is executed by an agent with label $\ell$, and called when the value of $C$ is 1.

\begin{center}
\fbox{
\begin{minipage}{7cm}

{\bf Procedure} {\tt LoseContact} $(C)$ \\

$T(\ell) \leftarrow (c_1 \dots c_{\lambda})$\\
$d \leftarrow 1$\\
$leading \leftarrow$ {\em false}\\
{\bf while} $C=1$ {\bf do}\\
\hspace*{1cm}$i \leftarrow 1$\\
\hspace*{1cm}{\bf while} ($C=1$ {\bf and} $i<\lambda$) {\bf do}\\
\hspace*{2cm}{\bf if} $c_i=1$ {\bf then} $move(N,d)$\\
\hspace*{2cm}{\bf else} stay inert for time $d$\\
\hspace*{2cm}$i \leftarrow i+1$\\
\hspace*{2cm}$check(C)$\\
\hspace*{1cm}$d \leftarrow 2d$\\
$j \leftarrow i-1$\\
{\bf if} $c_j=1$ {\bf then}\\
\hspace*{1cm}$leading \leftarrow true$\\
\hspace*{1cm}{\bf while} $C=0$ {\bf do}\\
\hspace*{2cm}$move(S,\frac{1}{2})$\\
\hspace*{2cm}$check(C)$

\end{minipage}
}
\end{center}

The second procedure is based on the following observation. Consider two points in the plane, $X$ and $A$, where $A$ is North of $X$ (in the sense defined in Section 2).
Suppose that the distance between $A$ and $X$ is $\rho$. Let $p$ be the vertical line containing point $A$, and let $B$ be the other point of line $p$ at distance $\rho$ from $X$.
(This includes the case when $A$ and $B$ coincide). Let $Z$ be the midpoint of the segment $AB$. Then $Z$ is on the same horizontal line as $X$ because the triangle $XAB$ is isoceles (cf. Fig. 1)\\

\begin{figure}[h] 
\begin{center}
\includegraphics[scale=0.6]{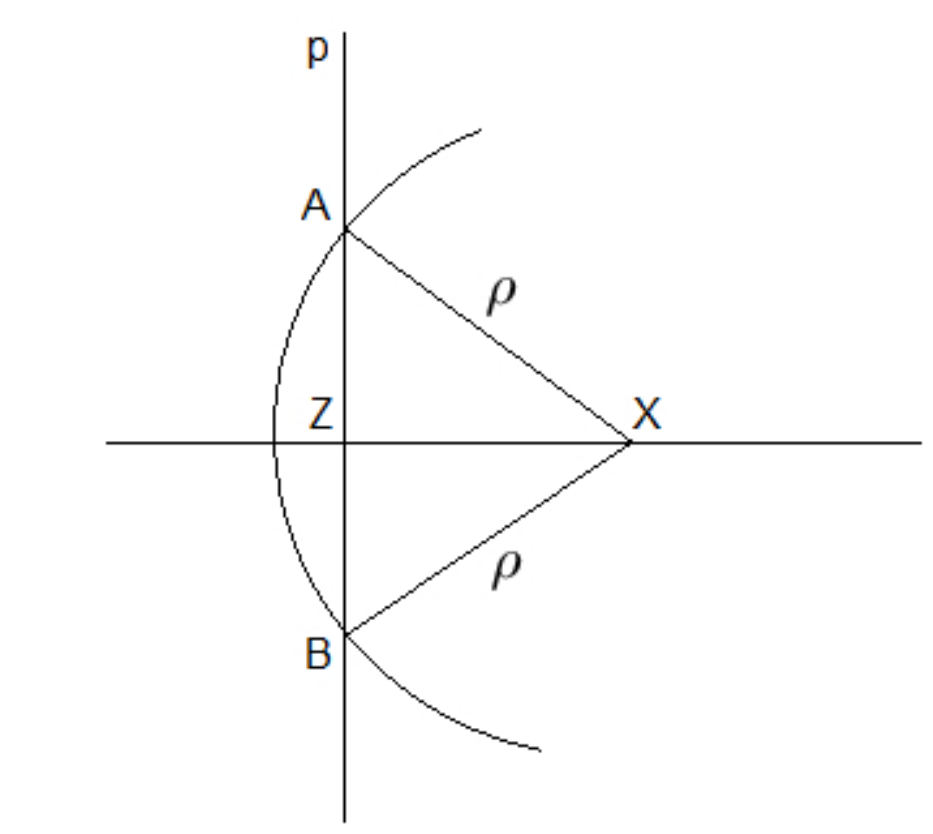} 
\end{center}
\caption{Geometric setting for procedure {\tt TriangleSearch} } 
\label{arbre non symetrique} 
\end{figure} 

The above observation can be used to construct the next procedure, called after executing Procedure {\tt LoseContact}, when (the center of) one of the agents, call it $a$, approximates a point $A$ North of the inert agent,
 with center $X$, such that the distance between $A$ and $X$ is $\rho$. Agent $a$ starts the procedure with the reading 1 of its sensor. It goes South with steps of length 1/2 (counted in the counter $t$), until the reading of its sensor becomes 0.
At this time its center approximates the other point $B$ at distance $\rho$ from $X$,  on the vertical line along which it travelled. Then the agent goes back (i.e., North) on this vertical line,
at distance $\lceil t/2 \rceil \cdot \frac{1}{2}$. Upon completing this move it reaches a point
which approximates the midpoint $Z$ of the segment $AB$. Now it goes horizontally, trying directions East and West in increasing leaps whose lengths double at each time, until meeting the inert agent. Since steps used by the agent, when it moved vertically, were sufficiently small, the approximations are good enough for the meeting to eventually occur.

\begin{center}
\fbox{
\begin{minipage}{7cm}

{\bf Procedure} {\tt TriangleSearch} $(C)$ \\

{\bf if} $leading=true$ {\bf then}\\
\hspace*{1cm}$t \leftarrow 0$\\
\hspace*{1cm}{\bf while} $C=1$ {\bf do}\\
\hspace*{2cm}$move(S,\frac{1}{2})$\\
\hspace*{2cm}$check(C)$\\
\hspace*{2cm}$t \leftarrow t+1$\\
\hspace*{1cm}$move(N, \lceil t/2 \rceil \cdot \frac{1}{2})$\\
\hspace*{1cm}$d \leftarrow 1$\\
\hspace*{1cm}{\bf repeat until} meeting\\
\hspace*{2cm}$move(E,d)$\\
\hspace*{2cm}$move(W,2d)$\\
\hspace*{2cm}$move(E,d)$\\
\hspace*{2cm}$d \leftarrow 2d$

\end{minipage}
}
\end{center}

Now our main algorithm for the binary model can be formulated succinctly as follows.
It is executed by an agent with  label $\ell$ (that intervenes in Procedure {\tt LoseContact}), with the understanding that when agents touch (i.e., get at distance 1), the algorithm is interrupted, as meeting is then accomplished.
We will prove that this will always occur by the end of the execution of the algorithm.

\begin{center}
\fbox{
\begin{minipage}{7cm}

{\bf Algorithm} {\tt MeetingWithBinarySensor}\\

$check(C)$\\
{\bf if} $C=0$ {\bf then} stay inert forever\\
{\bf else}\\
\hspace*{1cm} {\tt LoseContact} $(C)$\\
\hspace*{1cm} {\tt TriangleSearch} $(C)$

\end{minipage}
}
\end{center}

\vspace*{0.5cm}

\begin{theorem}
The meeting of two agents that are placed in the plane at an initial distance less than $\rho$ and execute Algorithm {\tt MeetingWithBinarySensor} in the binary model, occurs by the end of the execution of this algorithm.
The meeting occurs in time $O(\rho \log L)$ after the appearance of the later agent.
\end{theorem}

\begin{proof}
We first prove that agents get at distance at least $\rho$ during the execution of Procedure {\tt LoseContact}, i.e., that the ``{\bf while} $C=1$'' loop is eventually exited.
In the case of non-simultaneous start this is straightforward: when $d$ exceeds $\rho$, the active agent that always goes North gets at distance larger than $\rho$ in the vertical direction, at which point the distance between the agents exceeds $\rho$. If agents start simultaneously, then we show that they must get at distance at least $\rho$ at the latest at the end of the execution of the loop for the smallest $d$ exceeding $2\rho$.
Indeed, if this did not happen before, the agents are at distance smaller than $\rho$ in the vertical direction, at the beginning of this execution of the loop. Consider the first index $i$ for which the transformed labels
of the agents differ. Before processing this bit in this execution of the loop the agents are still at distance smaller than $\rho$ in the vertical direction. When processing this bit, one of the agents stays inert, and the other
one goes North at distance larger than $2\rho$, and hence gets at distance larger than $\rho$ from the other agent.

Moreover, agents realize that they got at distance at least $\rho$ at the reading of $check(C)$
which occurs after a move of one agent and a period of staying still of the other agent.
Consequently, symmetry between the agents is broken at this time. Call the agent after whose move this occurred, the {\em leading} agent.
This is the agent whose variable $leading$ is set to $true$ upon exiting the ``{\bf while} $C=1$'' loop.
(Notice that in the case of non-simultaneous start, the leading agent is the later one). This agent knows that it became leading.
The rest of the algorithm is executed only by the leading agent.

Let $p$ be the vertical line along which the leading agent travels during the execution of Procedure {\tt LoseContact}, and let $X$ be the center of the other agent.
Denote by $A$ the point on $p$ at distance $\rho$ from $X$, North of $X$, and denote by $B$ the point on $p$ at distance $\rho$ from $X$, South of $X$.
Upon completion of Procedure {\tt LoseContact} by the leading agent, this agent is on the line $p$ at distance at most $1/2$ from $A$. At the end of the execution of the {\bf while} loop of Procedure
{\tt TriangleSearch}, the agent is on the line $p$ at distance at most $1/2$ from $B$. Let $C$ be the point that the agent reaches after going at distance $\lceil t/2 \rceil \cdot \frac{1}{2}$ North on line $p$.
Point $C$ is at distance at most 1 from the center $Z$ of the segment $AB$. During the {\bf repeat} loop in Procedure {\tt TriangleSearch} the agent moves on the horizontal line containing point $C$, making
increasing leaps East and West. Hence it must eventually touch the other agent waiting in point $X$.

It remains to estimate the time of Algorithm {\tt MeetingWithBinarySensor}. The turn of the ``{\bf while} $C=1$'' loop in Procedure {\tt LoseContact}, for every parameter value $d$, takes time $O(d\lambda)$.
Since the parameter is doubled at each turn, and it is $\Theta(\rho)$ at the last turn, the entire loop takes time  $O(\rho\lambda)$. The remainder of Procedure {\tt LoseContact} is going South after the end of the loop.
This takes time $O(\rho)$. Moving vertically along line $p$ during Procedure {\tt TriangleSearch} takes time $O(\rho)$ as well. Horizontal moves ending this procedure must cover distance at most $\rho$, and they are
doubling every time, hence these moves also take time $O(\rho)$. It follows that the entire time of the algorithm is $O(\rho\lambda)=O(\rho \log L)$.
\end{proof}

Our last result shows that the complexity of Algorithm {\tt MeetingWithBinarySensor} cannot be improved in general, in the binary model.

\begin{theorem}
There is no deterministic algorithm for the binary model, guaranteeing meeting of any two agents starting at distance less than $\rho$ in the plane, and working in time $o(\rho \log L)$.
\end{theorem}

\begin{proof}
Consider a deterministic algorithm $\cal A$ for the binary model. Assume that $\rho>30$.
Define a tiling of the plane into pairwise disjoint squares of size $\rho/5$. Each tile includes its North and East edges.
For a fixed tile 0, enumerate the 8 neighboring tiles by numbers $1,\dots ,8$, starting with the tile North of tile 0 and going clockwise.
Let $A$ and $B$ be two tiles on the same horizontal line, separated by one tile. Call such tiles {\em conjugate}. Place one agent in any point of $A$, the other agent in any point of $B$, and start them simultaneously at time 0.

Suppose that algorithm $\cal A$ works in time at most $c\rho \lambda$, where $\lambda = \lceil \log L \rceil$ and $c<\frac{1}{24\log 9}$ is a constant.
Divide the time into consecutive segments of length $\rho/6$. If the agent is in some tile at the beginning of a time segment, then at the end of this segment it is either in the same or in one of the neighboring tiles.
The {\em behavior pattern} of an agent is the sequence $(a_1,\dots ,a_k)$ with terms $0,1,\dots ,8$, defined as follows.
If at the beginning of the $i$th segment the agent is in some tile 0, then $a_i=j$ if and only if the agent is in tile $j$ at the end of the segment. Since the algorithm works in time at most $c\rho \lambda$,
the length of the behavior pattern is at most $k=\lceil 6c\lambda \rceil$. Since $c<\frac{1}{24\log 9}$, we have $k<\frac{\log L}{\log 9}$, for sufficiently large $L$, and hence $9^k<L$.

As long as the agents are at distance less than $\rho$, and hence the reading of their sensors is always 1,
the actions of each agent depend only on its label. Agents start in conjugate tiles.
Consequently they start at distance at most $\frac{\sqrt{10}\rho}{5}$. Since time segments are of length $\rho/6$, the distance between the agents during the first time segment is always at most
$\frac{\sqrt{10}\rho}{5}+\frac{2\rho}{6} <\rho$. Hence the actions of each agent during the first time segment depend only on its label. There is a number in the set $\{0,1,\dots ,8\}$, such that for a set $S_1$ of at least $L/9$ labels,
if agents with distinct labels from this set start respectively in conjugate tiles that are given number 0, then
the tile numbers in which agents finish the first segment are the same. Hence agents with such labels finish the first time segment in conjugate tiles as well. Consequently, the distance between the agents during the second time segment is always at most
$\frac{\sqrt{10}\rho}{5}+\frac{2\rho}{6} <\rho$. Hence the behavior of each agent during the second time segment depends only on its label.
There exists a set $S_2 \subseteq S_1$ of size at least $|S_1|/9$, such that agents with labels in this set finish the second time segment in conjugate tiles as well.

Since $9^k<L$, by induction on the index of the time segment,  there are at least two labels such that agents starting in tiles $A$ and $B$ with these labels have the same behavior pattern.
Assign these labels to agents starting in tiles $A$ and $B$.
It follows that these agents are in conjugate tiles at the end of each time segment. Since time segments are of length $\rho/6$, we show that agents are always at distance
at least $\frac{\rho}{5}-\frac{\rho}{6}$. Indeed, suppose that during a time segment agents get at distance smaller than $\frac{\rho}{5}-\frac{\rho}{6}$. This could occur only
when both agents are inside the tile separating them at the beginning of the segment. However, during this time segment, each agent can penetrate this separating tile only at distance at most
$\rho/12$, because at the end of the segment they have to be again separated by this tile. Since at the beginning of the time segment agents were at distance at least $\frac{\rho}{5}$, this distance could not
decrease by more than $\frac{\rho}{6}$ during the time segment, which results in the distance at least $\frac{\rho}{5}-\frac{\rho}{6}=\frac{\rho}{30}>1$ at all times. Hence agents cannot meet.
This contradiction implies that the time of the algorithm must be larger than $c\rho \log L$, for sufficiently large $\rho$ and $L$,
which concludes the proof.
\end{proof}

We conclude this section with the observation that if agents start at distance $\alpha\rho$, for some constant $\alpha >1$ in the binary model, then sniffing does not help, i.e., the worst-case optimal meeting time is of the same
order of magnitude as without any sniffing ability.
 Indeed, consider two agents starting at distance $D=\alpha\rho$, for some constant $\alpha >1$. Let $\epsilon=\frac{\alpha -1}{2}$.
 In order to accomplish the meeting, the agents must first get
at distance $\rho(1+\epsilon)$. (It is enough to consider $\rho$ sufficiently large that $\rho(1+\epsilon)>1$).  This partial task must be accomplished in the model without any feedback, as the reading of sensors of both agents will be always 0 when it is performed. The partial task is thus equivalent to the task of approach \cite{DP} when agents start at distance $D'=1+\epsilon\rho$.
Hence the original task cannot be completed faster than in time $Opt(D')$, which denotes the optimal time of approach without sensors, when agents start at the initial distance $D'$.
The exact order of magnitude of  this optimal approach time is not known, but it is
 polynomial in $D'$ and in $\log L$. (It follows from \cite{DP} that this is the case even if agents are allowed some degree of asynchrony.) Since $D=\Theta(D')$, it follows that $Opt(D)=\Theta(Opt(D'))$. Consequently, the worst-case optimal meeting time
in the binary model with sniffing,  when agents start at distance $D=\alpha\rho$, for some constant $\alpha >1$, is of the same order of magnitude as if they start at the same initial distance and have no sensors.

\section{Conclusion}

We provided optimal algorithms for the task of meeting of two agents equipped with sniffing sensors, in two models of accuracy of these sensors.
In the monotone model it is assumed that sensors are perfectly accurate, i.e., they can  notice any change of distance between the agents, although they cannot measure the distance itself, nor recognize the direction in which the other agent is located. In the binary model, the sensor can only tell the agent if the other agent is close or far, for some threshold of closeness. We showed a separation between the two models:
while in the monotone model meeting is guaranteed in time proportional to the initial distance between the agents, in the binary model we showed that optimal meeting time is $\Theta(\rho \log L)$, where
$\rho$ is the sensing threshold and $L$ is the size of the label space.

In both models we assume that both agents travel at the same speed 1,
and this assumption is heavily used in our algorithms and their analysis. It is an interesting open problem how our results would change if agents moved in an asynchronous way, or at least were allowed the same
degree of asynchrony as in \cite{DP}, i.e., each agent moved at a constant speed, but possibly different from the other agent.

\bibliographystyle{plain}


\end{document}